\documentclass[draftcls,12pt,onecolumn,oneside]{IEEEtran}
\usepackage{pifont}
\usepackage{times,amsmath,color,amssymb,graphicx,epsfig,cite,psfrag,subfigure}
\usepackage{mathrsfs}
\usepackage{amsfonts}
\usepackage{algorithm}
\usepackage{algorithmic}
\usepackage{multirow}
\long\def\symbolfootnote[#1]#2{\begingroup%
\def\thefootnote{\fnsymbol{footnote}}\footnote[#1]{#2}\endgroup}

\newtheorem{pro}{\rm{\underline{\textbf{Problem}}}}

\newtheorem{thm}{\rm{\underline{\textbf{Theorem}}}}
\newtheorem{lem}{\rm{\underline{\textbf{Lemma}}}}

\begin{document}
\title{Ergodic Sum-Rate Maximization for Fading Cognitive Multiple Access Channels without Successive Interference Cancellation}
\author{~~~~~Xin Kang,~\IEEEmembership{Member,~IEEE}, ~Hon Fah Chong,~\IEEEmembership{Member,~IEEE}, \newline~Yeow-Khiang
Chia,~\IEEEmembership{Member,~IEEE}, ~Sumei
Sun,~\IEEEmembership{Senior Member,~IEEE},
\thanks{X. Kang, H.F. Chong, Y.-K Chia are with Institute for Infocomm Research, 1 Fusionopolis Way,
$\#$21-01 Connexis, South Tower, Singapore 138632 (E-mail: \{xkang,
hfchong, chiayk, sunsm\}@i2r.a-star.edu.sg).} } \markboth{SUBMITTED
FOR REVIEW} {} \maketitle

\begin{abstract}
In this paper, the ergodic sum-rate of a fading cognitive multiple
access channel (C-MAC) is studied, where a secondary network (SN)
with multiple secondary users (SUs) transmitting to a secondary base
station (SBS) shares the spectrum band with a primary user (PU). An
interference power constraint (IPC) is imposed on the SN to protect
the PU. Under such a constraint and the individual transmit power
constraint (TPC) imposed on each SU, we investigate the power
allocation strategies to maximize the ergodic sum-rate of a fading
C-MAC without successive interference cancellation (SIC). In
particular, this paper considers two types of constraints: (1)
\emph{average TPC {\rm {and}} average IPC}, (2) \emph{peak TPC {\rm
{and}} peak IPC}. For the first case, it is proved that the optimal
power allocation is dynamic time-division multiple-access (D-TDMA),
which is exactly the same as the optimal power allocation to
maximize the ergodic sum-rate of the fading C-MAC with SIC under the
same constraints. For the second case, it is proved that the optimal
solution must be at the extreme points of the feasible region. It is
shown that D-TDMA is optimal with high probability when the number
of SUs is large. Besides,  we show that, when the SUs can be sorted
in a certain order, an algorithm with linear complexity can be used
to find the optimal power allocation.

\end{abstract}
\begin{keywords} Cognitive Radio, Multiple Access Channel, Fading Channels, Spectrum Sharing, Ergodic Sum-Rate, Optimal Power Allocation, Non-convex Optimization.
\end{keywords}

\newpage

\section{Introduction}
The demand for frequency resources has dramatically increased due to
the explosive growth of wireless applications and services in recent
years. This poses a big challenge to the current fixed spectrum
allocation policy. On the other hand, a report published by Federal
Communications Commission (FCC) shows that the current scarcity of
spectrum resource is mainly due to the inflexible spectrum
regulation policy rather than the physical shortage of spectrum
\cite{FCC}. Most of the allocated frequency bands are
under-utilized, and the utilization of the spectrum varies in time
and space. Similar observations have also been made in other
countries. In particular, the spectrum utilization efficiency is
shown to be as low as $5\%$ in Singapore \cite{Sin5}. The compelling
need to improve the spectrum utilization and establish more flexible
spectrum regulations motivates the advent of cognitive radio (CR).
Compared to the traditional wireless devices, CR devices can greatly
improve the spectrum utilization by dynamically adjusting their
transmission parameters, such as transmit power, transmission rate
and the operating frequency. Recently, FCC has agreed to open the
licensed, unused television spectrum or the so-called white spaces
to the new, unlicensed, and sophisticatedly designed CR devices.
This milestone change of policy by the FCC indicates that CR is fast
becoming one of the most promising technologies for the future radio
spectrum utilization. This also motivates a wide range of research
in the CR area, including the research work done in this paper.

A popular model widely adopted in CR research is the \emph{spectrum
sharing} model. In a spectrum-sharing CRN, a common way to protect
primary users (PU) is to impose an interference power constraint
(IPC) at the secondary network, which requires the interference
received at PU receiver  to be below a prescribed threshold
\cite{Haykin2005}. Subject to such a IPC, the achievable rates of
Additive White Gaussian Noise (AWGN) channels were investigated in
\cite{Gastpar}. In \cite{Ghasemi}, the authors studied the ergodic
capacity of a single-user CRN under IPC in different fading
environment. In \cite{huang2009spectrum}, the authors studied the
outage performance of such a single-user spectrum-sharing CRN under
a IPC. In \cite{Leila09}, the authors studied the capacity and power
allocation for a spectrum-sharing fading CRN under both peak and
average IPC.  In \cite{kangTWC}, the optimal power allocation
strategies to achieve the ergodic and outage capacity for a
spectrum-sharing fading CRN under different combinations of the
transmit power constraint (TPC) and the IPC were investigated.
However, the aforementioned works only focused on the point-to-point
secondary networks. In \cite{PCheng2005}, from an information
theoretic perspective, the authors investigated the achievable rate
region of a Gaussian C-MAC. In \cite{zhang2008joint} and
\cite{zhang2009cognitive}, the authors investigated the optimal
power allocation strategies for AWGN cognitive multiple access
channels (C-MAC). In \cite{zhang2009ergodic}, the authors
investigated the ergodic sum capacity for a fading C-MAC with
multiple PUs. 
In \cite{XKang2010}, the authors
studied the outage capacity region for a fading C-MAC. However, in
these works, successive interference cancellation (SIC) decoders are
assumed to be available, and thus no mutual interference among the
secondary users (SU) is considered. Different from the
aforementioned works, in this paper, we study the ergodic sum-rate
and the corresponding optimal power allocation strategies of a
fading C-MAC without SIC. Compared with the previous studies with
SIC, the problem studied in this paper is much harder due to the
existence of the mutual interference among SUs, which makes the
problem a nonlinear, nonconvex constrained optimization problem.

Another line of related research
\cite{HanlyPart1,Wyu2004,AGjendemsj2008,Inaltekin2012} focused on
the sum-rate maximization for MAC under non-CR setting (without
IPC). In \cite{HanlyPart1}, the authors investigated the ergodic
capacity region and its optimal power allocation for the fading MAC.
In \cite{Wyu2004}, the authors proposed the iterative water-filling
algorithm to maximize the sum-rate of a multiple-input multiple-out
(MIMO) MAC with SIC under individual power constraints. For sum-rate
maximization of MAC without SIC, in \cite{AGjendemsj2008}, the
authors were able to show the optimality of the binary power
allocation for a two-user network. For arbitrary users, the authors
only numerically illustrated the optimality of binary power
allocation. While in \cite{Inaltekin2012}, the authors analytically
proved that binary power allocation is optimal for any number of
users in terms of maximizing the sum-rate of the MAC without SIC.
Compared with these works, the problem studied in this paper is more
challenging due to the existence of the IPC, which changes the
properties of the optimal power allocation. It is shown that binary
power allocation is no longer optimal for our problem.

The main contribution and the key results of this paper are listed
as follows:
\begin{itemize}
  \item We investigate the optimal power allocation strategies to maximize the ergodic sum-rate of a
fading C-MAC without SIC under both TPC and IPC. In particular, we
consider two types of constraints : (1) \emph{average TPC {\rm
{and}} average IPC}, (2) \emph{peak TPC {\rm {and}} peak IPC}.

 \item For the \emph{average TPC {\rm
{and}} average IPC} case,  we prove that the optimal power
allocation is dynamic time-division multiple-access (D-TDMA), which
is exactly the same as the optimal power allocation given in
\cite{zhang2009ergodic} to maximize the ergodic sum-rate of the
fading C-MAC with SIC under the same constraints.

  \item  For the \emph{peak TPC {\rm {and}} peak IPC} case, we prove that the optimal
solution must be at the extreme points of the feasible region. We
show that D-TDMA is optimal when a  certain condition is satisfied.
It is also shown that D-TDMA is optimal with high probability when
the number of SUs is large. Thus, we can solve the problem by
searching the extreme points of the feasible region when the number
of SUs is small, and by applying the D-TDMA scheme when the number
of SUs is large.

\item For the \emph{peak TPC {\rm {and}} peak IPC} case, we show that when the SUs can be sorted in a certain
order, an algorithm with linear complexity can be developed to find
the optimal power allocation of our problem.

  \item  For the \emph{peak TPC {\rm {and}} peak IPC} case, we also show by simulations that the optimal power allocation to maximize the ergodic sum-rate of the
fading C-MAC with SIC, which we refer to as \emph{SIC-OP}, can be
used as a good suboptimal power allocation for our problem. It is
shown by simulations that SIC-OP is optimal or near-optimal for our
problem when the D-TDMA is not optimal.
\end{itemize}

The rest of the paper is organized as follows. The system model and
power constraints are described in Section \ref{SystemModel}. The
optimal power allocation strategies to maximize the ergodic sum-rate
of the fading C-MAC without SIC are studied in Section
\ref{Sec-OPA}. Then, the simulation results are presented and
analyzed in Section \ref{Sec-simulation}. Section \ref{conclusion}
concludes the paper.

\section{System Model and Power Constraints}\label{SystemModel}
\subsection{System Model}
In this paper, we consider a spectrum sharing CR network consists of
one PU and a $K$-user secondary multiple access network. The
communication links between each SU and the PU receiver (PU-Rx) are
referred as the interference links. The links between the SUs and
the secondary base station (SBS) are referred as the secondary
links. For the convenience of exposition, all the channels involved
are assumed to be block-fading (BF) \cite{Biglieri1998}, i.e., the
channels remain constant during each transmission block, but
possibly change from one block to another. As shown in
Fig.\ref{model}, the channel power gain of the interference link between SU-$i$ and the PU is denoted by $g_i$. 
The channel power gain of the secondary link between SU-$i$ and the
SBS is denoted as $h_{i}$. All these channel power gains are assumed
to be independent and identically distributed (i.i.d.) random
variables (RVs) each having a continuous probability density
function (PDF). All the channel state information (CSI) is assumed
to be perfectly known at both SUs. CSI of the secondary links can be
obtained at SUs by the classic channel training, estimation, and
feedback mechanisms. CSI of the interference links between SUs and
primary receivers can be obtained at SUs via the cooperation of the
primary receivers. The noise at SBS is assumed to be circular
symmetric complex Gaussian variable with zero mean and variance
$\sigma^2$ denoted by ${\mathcal C}{\mathcal N}(0,\sigma^2)$.

\subsection{Power Constraints}\label{Sec-Powerconstraints}
In this paper, we denote the transmit power of SU-$i$ as $P_i$, then
the instantaneous interference received at PU-Rx from SU-$i$ is
$g_iP_i$. Then, the average and peak \emph{interference power
constraint} (IPC) can be described as
\begin{align}
\mbox{Average IPC:}~~&\mathbb{E}\left[\sum_{i=1}^K g_iP_i\right]
\le I^{av}, \label{Con-AIPC}\\
\mbox{Peak IPC:}~~&\sum_{i=1}^K g_iP_i\le I^{pk},\label{Con-PIPC}
\end{align}
where $I^{av}$ denotes the limit of average received interference at
the PU, and $I^{pk}$ denotes the maximum instantaneous interference
that the PU can tolerate. $\mathbb{E}[\cdot]$ denotes the
statistical expectation over all the involved fading channel power
gains. Usually, the average IPC is used to guarantee the long-term
QoS of the PU when it provides delay-insensitive services. When the
service provided by the PU has an instantaneous QoS requirement, the
peak IPC is usually adopted.

In this paper, we also consider the \emph{transmit power constraint}
(TPC) imposed at each SU.  Same as the IPC, two types (both average
and peak) of TPC are considered here. Let $P_i^{av}$ and $P_i^{pk}$
be the average and peak transmit power limit of SU-$i$,
respectively. Then, the average and peak TPC can be described as
\begin{align}
\mbox{Average TPC:}~~&\mathbb{E}\left[P_i\right] \le P_i^{av},~\forall i, \label{Con-ATPC}\\
\mbox{Peak TPC:}~~& P_i\le P_i^{pk},~\forall i,\label{Con-PTPC}
\end{align}
where $\mathbb{E}[\cdot]$ denotes the statistical expectation over
all the involved fading channel power gains. The peak power
limitation is usually due to the nonlinearity of power amplifiers in
practice. The average TPC is usually imposed to meet a long-term
transmit power budget.


\section{Ergodic Sum-rate Maximization for Fading
C-MAC without SIC} \label{Sec-OPA}
Without SIC decoders available at
the SBS, the instantaneous transmission rate of each SU is given by
\begin{align}
r_i=\ln\left(1+\frac{h_{i}p_i}{\sum_{j=1,j\neq i}^K
h_{j}P_j+\sigma^2}\right), ~\forall i. \label{Eq-ri}
\end{align}

For BF channels, ergodic rate is defined as the maximum achievable
rate averaged over all the fading blocks. Then, the ergodic sum-rate
of the fading C-MAC considered in this paper can be written as
\begin{align}
\mathbb{E}\left[\sum_{i=1}^K \ln\left(1+\frac{
h_{i}P_i}{\sigma^2+\sum_{j=1,j\neq i}^K h_{j}P_j}\right)\right].
\end{align}

In the following, we study power allocation strategy to maximize the
ergodic sum-rate of the fading C-MAC subject to the power
constraints given in Section \ref{Sec-Powerconstraints}.
\subsection{Average TPC and Average IPC}
Under average TPC and average IPC, the optimal power allocation to
maximize the ergodic sum-rate of the fading C-MAC can be obtained by
solving the following optimization problem:
\begin{pro}\label{Problem-AverageAverage}
\begin{align}
\max_{P_i\ge 0,\forall i}~~ &\mathbb{E}\left[\sum_{i=1}^K
\ln\left(1+\frac{
h_{i}P_i}{\sigma^2+\sum_{j=1,j\neq i}^K h_{j}P_j}\right)\right],\\
\mbox{s.t.}~~~~ &\eqref{Con-AIPC},~\eqref{Con-ATPC}.
\end{align}
\end{pro}

It is not difficult to observe that Problem
\ref{Problem-AverageAverage} is a non-convex optimization problem.
Thus, we cannot solve it by the standard convex optimization
techniques. To solve Problem \ref{Problem-AverageAverage}, we first
look at the following problem.
\begin{pro}\label{Problem-AverageAverage-SIC}
\begin{align}
\max_{P_i\ge 0,\forall i}~~ &\mathbb{E}\left[\ln\left(1+\sum_{i=1}^K\frac{h_{i}P_i}{\sigma^2}\right)\right],\\
\mbox{s.t.}~~~~ &\eqref{Con-AIPC},~\eqref{Con-ATPC}.
\end{align}
\end{pro}
Problem \ref{Problem-AverageAverage-SIC} gives the ergodic sum-rate
for fading C-MAC with SIC, and it has been studied in
\cite{zhang2009ergodic}.  It is shown in \cite{zhang2009ergodic}
(Lemma 3.1) that the optimal solution of Problem
\ref{Problem-AverageAverage-SIC} is: at most one user is allowed to
transmit in each fading block. Based on this fact, we obtain the
following theorem.

\begin{thm}\label{Thm-1}
The optimal solution of Problem \ref{Problem-AverageAverage} is the
same as that of Problem \ref{Problem-AverageAverage-SIC}.
\end{thm}

\begin{proof}
It is observed that the constraints of Problem
\ref{Problem-AverageAverage} and Problem
\ref{Problem-AverageAverage-SIC} are exactly the same. Thus, the
feasible sets of Problem \ref{Problem-AverageAverage} and Problem
\ref{Problem-AverageAverage-SIC} are the same. Now, suppose
$\boldsymbol{P}=[P_1~ P_2~ \cdots ~P_K]^T$ is a feasible solution of
Problem \ref{Problem-AverageAverage}. The rest of the proof consists
of two steps.

\emph{Step $1$:} Since $\boldsymbol{P}$ is a feasible solution of
Problem \ref{Problem-AverageAverage}, it is also a feasible solution
of Problem \ref{Problem-AverageAverage-SIC}. Now, we show that the
value of the objective function of Problem
\ref{Problem-AverageAverage-SIC} under $\boldsymbol{P}$ is an
upper-bound of that of Problem \ref{Problem-AverageAverage} under
the same $\boldsymbol{P}$, i.e., $\mathbb{E}\left[\sum_{j=1}^K
\ln\left(1+\frac{ h_{j}P_j}{\sigma^2+\sum_{i=1,i\neq j}^K
h_{i}P_i}\right)\right] \le
\mathbb{E}\left[\ln\left(1+\sum_{i=1}^K\frac{h_{i}P_i}{\sigma^2}\right)\right]$.
Since the expectation operation is linear, it is equivalent to show
that $\sum_{j=1}^K \ln\left(1+\frac{
h_{j}P_j}{\sigma^2+\sum_{i=1,i\neq j}^K h_{i}P_i}\right)\le
\ln\left(1+\sum_{i=1}^K\frac{h_{i}P_i}{\sigma^2}\right)$, which is
given below.
\begin{align}
\ln\left(1+\sum_{i=1}^K\frac{h_{i}P_i}{\sigma^2}\right)=&\ln\left(\frac{\sigma^2+\sum_{i=1}^K
h_{i}P_i}{\sigma^2}\right)\nonumber
\\=&\ln\left[\left(\frac{\sigma^2+\sum_{i=1}^K
h_{i}P_i}{\sigma^2+\sum_{i=2}^K
h_{i}P_i}\right)\left(\frac{\sigma^2+\sum_{i=2}^K
h_{i}P_i}{\sigma^2+\sum_{i=3}^K
h_{i}P_i}\right)\cdots\left(\frac{\sigma^2+\sum_{i=K}^K
h_{i}P_i}{\sigma^2}\right)\right]\nonumber
\\\stackrel{a}{=}&\sum_{j=1}^K
\ln\left(\frac{\sigma^2+\sum_{i=j}^K
h_{i}P_i}{\sigma^2+\sum_{i=j+1}^K h_{i}P_i}\right)\nonumber
\\=&\sum_{j=1}^K \ln\left(1+\frac{
h_{j}P_j}{\sigma^2+\sum_{i=j+1}^K h_{i}P_i}\right)\nonumber
\\\stackrel{b}{\ge}&\sum_{j=1}^K \ln\left(1+\frac{
h_{j}P_j}{\sigma^2+\sum_{i=1,i\neq j}^K h_{i}P_i}\right),
\end{align}
where we introduce a dumb item $\sum_{i=K+1}^K h_{i}P=0$ in the
equality ``a'' for notational convenience. The inequality ``b''
follows from the fact that $\sum_{i=1,i\neq j}^K h_{i}P_i \ge
\sum_{i=j+1}^K h_{i}P_i, \forall j$.

\emph{Step $2$:} Now, we show that the optimal solution of Problem
\ref{Problem-AverageAverage} is the same as that of Problem
\ref{Problem-AverageAverage-SIC}. Since it is proved in
\cite{zhang2009ergodic} (Lemma 3.1) that the optimal solution of
Problem \ref{Problem-AverageAverage-SIC} is: at most one user is
allowed to transmit in each fading block. It is easy to observe that
the optimal solution of Problem \ref{Problem-AverageAverage-SIC} is
a feasible solution of Problem \ref{Problem-AverageAverage}. Since
we have shown in Step $1$ that Problem
\ref{Problem-AverageAverage-SIC} provides an upper-bound of Problem
\ref{Problem-AverageAverage} for the same $\boldsymbol{P}$. Thus, it
is easy to observe that the optimal solution of Problem
\ref{Problem-AverageAverage} must be the same as that of Problem
\ref{Problem-AverageAverage-SIC}, which is: at most one user is
allowed to transmit in each fading block.
\end{proof}

Since in Theorem \ref{Thm-1}, we have shown that the optimal
solution of Problem \ref{Problem-AverageAverage} is the same as that
of Problem \ref{Problem-AverageAverage-SIC}. Thus, the optimal power
allocation strategies for Problem \ref{Problem-AverageAverage} can
be obtained in the same way as \cite{zhang2009ergodic}. Interested
readers can refer to Lemma  3.1 and 3.2 in \cite{zhang2009ergodic}
for details.

\subsection{Peak TPC and Peak IPC}
Under peak TPC and peak IPC, the optimal power allocation to
maximize the ergodic sum-rate of the fading C-MAC can be obtained by
solving the following optimization problem:
\begin{pro}\label{Problem-Peakpeak}
\begin{align}
\max_{P_i\ge 0, \forall i}~~ &\mathbb{E}\left[\sum_{i=1}^K
\ln\left(1+\frac{
h_{i}P_i}{\sigma^2+\sum_{j=1,j\neq i}^K h_{j}P_j}\right)\right],\\
\mbox{s.t.}~~ & \eqref{Con-PIPC},~\eqref{Con-PTPC}.
\end{align}
\end{pro}
Since all the constraints involved are instantaneous power
constraints, Problem \ref{Problem-Peakpeak} can be decomposed into a
series of identical subproblems each for one fading state, which is
\begin{pro}\label{Problem-Peakpeak2}
\begin{align}
\max_{P_i\ge 0, \forall i}~~ &\sum_{i=1}^K \ln\left(1+\frac{
h_{i}P_i}{\sigma^2+\sum_{j=1,j\neq i}^K h_{j}P_j}\right),\label{Eq-PeakObj}\\
\mbox{s.t.}~~ & \eqref{Con-PIPC},~\eqref{Con-PTPC}.
\end{align}
\end{pro}

It can be verified that Problem \ref{Problem-Peakpeak2} is
non-convex. Thus, we cannot solve it directly by the standard convex
optimization techniques. To solve Problem \ref{Problem-Peakpeak2},
we first investigate its properties.

\begin{lem}\label{Proposition-boundary}
The optimal solution $\boldsymbol{P}^*$ of Problem
\ref{Problem-Peakpeak2} must be at the boundary of the feasible
region of Problem \ref{Problem-Peakpeak2}.
\end{lem}

\begin{proof}
This can be proved by contradiction. Suppose the optimal solution
$\boldsymbol{P}^*$ of Problem \ref{Problem-Peakpeak2} is in the
interior of the feasible region, i.e., $0<P_i^*<P_i^{pk}, \forall i$
and $\sum_{i=1}^K g_iP^*_i< I$.

Now, we look at the power allocation $P_n$ of SU-$n$. For
convenience, we denote \eqref{Eq-PeakObj} as
$f\left(\boldsymbol{P}\right)$. Then, $f\left(\boldsymbol{P}\right)$
can be rewritten  as
\begin{align}
f\left(\boldsymbol{P}\right)=\ln\left(1+\frac{
h_{n}P_n}{\sigma^2+\sum_{j=1,j\neq n}^K
h_{j}P_j}\right)+\sum_{i=1,i\neq n}^K\ln\left(1+\frac{
h_{i}P_i}{\sigma^2+\sum_{j=1,j\neq i}^K h_{j}P_j}\right).
\end{align}
Taking the derivative of $f\left(\boldsymbol{P}\right)$ with respect
to $P_n$, we have
\begin{align}
\frac{\partial f\left(\boldsymbol{P}\right)}{\partial
P_n}&=\frac{1}{1+\frac{ h_{n}P_n}{\sigma^2+\sum_{j=1,j\neq n}^K
h_{j}P_j}}*\frac{h_{n}}{\sigma^2+\sum_{j=1,j\neq n}^K h_{j}P_j} \nonumber\\
&+\sum_{i=1,i\neq
n}^K\frac{1}{1+\frac{h_iP_i}{\sigma^2+\sum_{j=1,j\neq i}^K
h_{j}P_j}}*\left(-\frac{h_iP_i}{\left(\sigma^2+\sum_{j=1,j\neq i}^K
h_{j}P_j\right)^2}\right)*h_n \nonumber\\
&=\frac{h_n}{\sigma^2+\sum_{j=1}^K h_{j}P_j}-\sum_{i=1,i\neq
n}^K\frac{ h_iP_i h_n}{\left(\sigma^2+\sum_{j=1}^K
h_{j}P_j\right)\left(\sigma^2+\sum_{j=1,j\neq i}^K
h_{j}P_j\right)}\nonumber\\
&=\frac{h_n}{\sigma^2+\sum_{j=1}^K h_{j}P_j}\left(1-\sum_{i=1,i\neq
n}^K\frac{ h_iP_i}{\left(\sigma^2+\sum_{j=1,j\neq i}^K
h_{j}P_j\right)}\right).
\end{align}

It is observed that $q(P_n)\triangleq1-\sum_{i=1,i\neq n}^K\frac{
h_iP_i h_n}{\left(\sigma^2+\sum_{j=1,j\neq i}^K h_{j}P_j\right)}$ is
a strictly increasing function with respect to $P_n$. Then, the
solution to $q(P_n)=0$ is unique. Consequently, the solution to
$\frac{\partial f\left(\boldsymbol{P}\right)}{\partial P_n}=0$ is
also unique since $\frac{h_n}{\sigma^2+\sum_{j=1}^K h_{j}P_j}$ is
strictly positive. Denote the solution of $\frac{\partial
f\left(\boldsymbol{P}\right)}{\partial P_n}=0$ as $\tilde{P}_n$, and
we refer to $\tilde{P}_n$ as the turning point. Then, on the left
side of the turning point, $\frac{\partial
f\left(\boldsymbol{P}\right)}{\partial P_n}$ is always negative,
thus $f(0)>f(P_n), \forall P_n \in \left[0,\tilde{P}_n\right]$. On
the right side of the turning point, $\frac{\partial
f\left(\boldsymbol{P}\right)}{\partial P_n}$ is always positive,
thus $f(P_n)<f(P_n^{mx}), \forall P_n \in
\left[\tilde{P}_n,P_n^{mx}\right]$, where
$P_n^{mx}=\min\left\{P_n^{pk},\left(I-\sum_{i=1,i\neq n}^K
g_iP^*_i\right)\big/{h_n}\right\}$. Thus, it is clear that the value
of $f\left(\boldsymbol{P}\right)$ can be increased by moving $P_n$
to the boundary. This contradicts with our assumption that
$\boldsymbol{P}^*$ is the optimal solution. Thus, Lemma
\ref{Proposition-boundary} is proved.
\end{proof}

Based on the result of Lemma \ref{Proposition-boundary}, we are able
to obtain the following theorem.
\begin{thm}\label{Proposition-extreme}
The optimal solution $\boldsymbol{P}^*$ of Problem
\ref{Problem-Peakpeak2} must be at the extreme point of the feasible
region of Problem \ref{Problem-Peakpeak2}, i.e., at most one user's
power allocation is fractional.
\end{thm}

\begin{proof}
Suppose the optimal solution is $\boldsymbol{P}^*$. Thus, if
$\sum_{i=1}^K g_iP^*_i< I$, based the results of Lemma
\ref{Proposition-boundary}, it is clear that $P_i^*, \forall i$ is
either equal to $0$ or $P_i^{pk}$. Thus, there is no fractional
user.

Now, we consider the case that $\sum_{i=1}^K g_iP^*_i=I$. Suppose
$P^*_1$ and $P^*_2$ are fractional, i.e., $0<P^*_i<P^{pk}_i,\forall
i\in\left\{1,2\right\}$. The interference constraint can be
rewritten as $g_1P^*_1+g_2P^*_2+\sum_{i=3}^Kg_iP^*_i=I$. For
convenience, we define $Q\triangleq I-\sum_{i=3}^Kg_iP^*_i$.

First, we consider the case that $\frac{h_1}{g_1}>\frac{h_2}{g_2}$.
Under this assumption, we write $P_2^*$ as $P^*_2=(Q-g_1P^*_1)/g_2$.
For convenience, we denote \eqref{Eq-PeakObj} as
$f\left(\boldsymbol{P}\right)$. Then, $f\left(\boldsymbol{P}\right)$
can be rewritten  as
\begin{align}
f\left(\boldsymbol{P}^*\right)&=\ln\left(1+\frac{
h_{1}P^*_1}{\sigma^2+\sum_{j=3}^K
h_{j}P^*_j+h_{2}(Q-g_1P^*_1)/g_2}\right)+\ln\left(1+\frac{
h_{2}(Q-g_1P^*_1)/g_2}{\sigma^2+\sum_{j=3}^K
h_{j}P^*_j+h_{1}P^*_1}\right)\nonumber\\&+\sum_{i=3}^K\ln\left(1+\frac{
h_{i}P^*_i}{\sigma^2+\sum_{j=3, j\neq i}^K
h_{j}P^*_j+h_{1}P^*_1+h_{2}(Q-g_1P^*_1)/g_2}\right). \label{Eq-fP}
\end{align}

For notation convenience, define $C\triangleq\sigma^2+\sum_{j=3}^K
h_{j}P^*_j$ and $D_i\triangleq\sigma^2+\sum_{j=3, j\neq i}^K
h_{j}P^*_j$, then \eqref{Eq-fP} can be rewritten as
\begin{align}
f\left(\boldsymbol{P}^*\right)&=\ln\left(1+\frac{
h_{1}P^*_1}{C+h_{2}(Q-g_1P^*_1)/g_2}\right)+\ln\left(1+\frac{
h_{2}(Q-g_1P^*_1)/g_2}{C+h_{1}P^*_1}\right)\nonumber\\&+\sum_{i=3}^K\ln\left(1+\frac{
h_{i}P^*_i}{D_i+h_{1}P^*_1+h_{2}(Q-g_1P^*_1)/g_2}\right).
\end{align}

Taking the derivative of $f\left(\boldsymbol{P}^*\right)$ with
respect to $P^*_1$, we have
\begin{align}
\frac{\partial f\left(\boldsymbol{P}^*\right)}{\partial
P^*_1}&=\frac{1}{C+\frac{h_2}{g_2}Q+\left(h_1-\frac{h_2g_1}{g_2}\right)P^*_1}
\left(\frac{h_1C+\frac{h_1h_2}{g_2}Q}{C+\frac{h_2}{g_2}Q-\frac{h_2g_1}{g_2}P^*_1}
-\frac{\frac{h_2g_1}{g_2}C+\frac{h_1h_2}{g_2}Q}{C+h_1P^*_1}\right.
\nonumber
\\&\left.-\sum_{i=3}^K\frac{\left(h_1-\frac{h_2g_1}{g_2}\right)h_iP^*_i}{D_i+\frac{h_2}{g_2}Q+\left(h_1-\frac{h_2g_1}{g_2}\right)P^*_1}\right),\label{eq-23}
\end{align}
Since $\frac{h_1}{g_1}>\frac{h_2}{g_2}$, $\frac{\partial
f\left(\boldsymbol{P}^*\right)}{\partial P^*_1}$ is a strictly
increasing function with respect to $P^*_1$. Thus, the solution to
$\frac{\partial f\left(\boldsymbol{P}^*\right)}{\partial P^*_1}=0$
is unique. Denote the solution of $\frac{\partial
f\left(\boldsymbol{P}^*\right)}{\partial P^*_1}=0$ as
$\tilde{P}^*_1$, and we refer to $\tilde{P}^*_1$ as the turning
point. Then, on the left side of the turning point, $\frac{\partial
f\left(\boldsymbol{P}^*\right)}{\partial P^*_1}$ is always negative,
thus $f(0)>f(P_1), \forall P_1 \in \left[0,\tilde{P}^*_1\right]$. On
the right side of the turning point, $\frac{\partial
f\left(\boldsymbol{P}^*\right)}{\partial P^*_1}$ is always positive,
thus $f(P_1)<f(P_1^{pk}), \forall P_1 \in
\left[\tilde{P}^*_1,P_1^{pk}\right]$. Thus, it is clear that the
value of $f\left(\boldsymbol{P}^*\right)$ can be increased by moving
$P^*_1$ to $0$ or $P_1^{pk}$. This contradicts with our assumption
that $\boldsymbol{P}^*$ is the optimal solution.

Now, we consider the case that $\frac{h_1}{g_1}<\frac{h_2}{g_2}$.
For this case, we can write $P_1^*$ as $P_1^*=(Q-g_2P^*_2)/g_1$.
Then, using the same approach, we can show that the value of
$f\left(\boldsymbol{P}^*\right)$ can be increased by moving $P^*_2$
to $0$ or $P_2^{pk}$.

Combining the above results, it is observed that at most one user's
power allocation can be fractional. Theorem
\ref{Proposition-extreme} is thus proved.
\end{proof}

Based on Theorem \ref{Proposition-extreme}, we can easily find the
optimal solution $\boldsymbol{P}^*$ of Problem
\ref{Problem-Peakpeak2} by searching the extreme points when the
number of SUs is relatively small. However, when the number of SUs
is large, this scheme may not be practical due to the high computing
complexity. Fortunately, we are able to to show that with high
probability, the optimal solution is D-TDMA when the number of SUs
is large. This is given in Theorem \ref{Proposition-HighProb}.

To prove Theorem \ref{Proposition-HighProb}, we need the following
lemma.

\begin{lem}\label{Proposition-LargeNumebrSU}
The optimal solution $\boldsymbol{P}^*$ of Problem
\ref{Problem-Peakpeak2} is $P_k^*= \min\left\{P_k^{pk},
\frac{I^{pk}}{g_k}\right\}$  where $k=\mbox{argmax}_i~
\min\left\{h_iP_i^{pk}, \frac{h_i}{g_i}I^{pk}\right\}$, and
$P_i^*=0, \forall i\neq k$, if the condition
$\ln\left(1+h_{k}P_{k}^*/\sigma^2\right)\ge 1$ holds.
\end{lem}

\begin{proof}
It is shown that in \cite{Inaltekin2012} (Theorem 4), the optimal
solution for Problem \ref{Problem-Peakpeak2} without IPC is
single-user transmission if at least one user satisfies
$\ln\left(1+h_{i}P_{i}/\sigma^2\right)\ge 1$, and the channel is
assigned to the user with the largest $h_iP_i^{pk}$ at the current
fading block. Our proof is mainly based on this result.

Define $T_i\triangleq\min\left\{P_i^{pk},
\frac{I^{pk}}{g_i}\right\}$. Suppose there exists at least one user
satisfying the condition $\ln\left(1+h_{i}T_{i}/\sigma^2\right)\ge
1$. Since the condition $\ln\left(1+h_{i}T_{i}/\sigma^2\right)\ge 1$
holds, it follows from \cite{Inaltekin2012}  that the objective
function of Problem \ref{Problem-Peakpeak2} is maximized when only
one user transmits in each fading block. When there  is only one
user transmitting, the objective function of Problem
\ref{Problem-Peakpeak2} reduces to
$\ln\left(1+h_{i}P_{i}/\sigma^2\right)$, and the constraints reduces
to $P_i \le P_i^{pk}$ and $g_i P_i \le I^{pk}$. Clearly, the user
with the largest $h_i T_i$ will maximize the objective function.
Thus, the optimal allocation is $P_k^*= T_k$ where
$k=\mbox{argmax}_i~ h_i T_i$, and $P_i^*=0, \forall i\neq k$. Lemma
\ref{Proposition-LargeNumebrSU} is thus proved.
\end{proof}

\begin{thm}\label{Proposition-HighProb}
When the number of SUs is large, with high probability, the optimal
solution of Problem \ref{Problem-Peakpeak} is D-TDMA, i.e., one user
transmitting in each fading block.
\end{thm}
\begin{proof}
From Lemma \ref{Proposition-LargeNumebrSU}, it is known that if
there exists at least one user satisfying the condition
$\ln\left(1+h_{i}T_{i}/\sigma^2\right)\ge 1$ where
$T_i\triangleq\min\left\{P_i^{pk}, \frac{I^{pk}}{g_i}\right\}$, the
optimal solution of Problem \ref{Problem-Peakpeak} is dynamic TDMA.
Since all the channel power gains are i.i.d. , the probability of no
user satisfying $\ln\left(1+h_{i}T_{i}/\sigma^2\right)\ge 1$ is
$1-\left(\mbox{Prob}\left\{\ln\left(1+h_{i}T_{i}/\sigma^2\right)<
1\right\}\right)^K$. It is observed that this probability is a
monotonic increasing function with respect to $K$. Thus, when the
number of SUs is large, with high probability, the condition will
hold. Theorem \ref{Proposition-HighProb} is thus proved.
\end{proof}

Based on these results, we can solve Problem \ref{Problem-Peakpeak}
by searching the extreme points of the feasible region when the
number of SUs is small, and by applying the D-TDMA scheme when the
number of SUs is large. Readers may be interested in the number of
SUs that is required to make the D-TDMA scheme optimal. We have
investigated this issue in the simulation part given in Section
\ref{Sec-simulation}. Please note that the condition given in Lemma
\ref{Proposition-LargeNumebrSU} is only a sufficient condition. In
practice, the probability that D-TDMA is optimal is higher than
$1-\left(\mbox{Prob}\left\{\ln\left(1+h_{i}T_{i}/\sigma^2\right)<
1\right\}\right)^K$. For the commonly used parameters, D-TDMA can
achieve a near-optimal performance when the number of SUs is
moderate (such as $K=5$).

In the above, we have presented the approach to solve Problem
\ref{Problem-Peakpeak2} in general.  In the following, we show that
if the SUs can be sorted in certain order according to their channel
power gains, a simple algorithm with linear time complexity can be
developed to solve Problem \ref{Problem-Peakpeak2}.

\begin{thm}\label{Proposition-specialcase}
If the SUs can be sorted in the following order:
$h_1>h_2>\cdots>h_K$ and
$\frac{g_1}{h_1}<\frac{g_2}{h_2}<\cdots<\frac{g_K}{h_K}$. Then,
there exists an optimal solution, for any two users indexed by $m$
and $n$, if $m<n$, their power allocation satisfies $P^*_m \ge
P^*_n$.

\end{thm}
\begin{proof} 
Assume that the users can be sorted in the following order:
$h_1>h_2>\cdots>h_K$ and
$\frac{g_1}{h_1}<\frac{g_2}{h_2}<\cdots<\frac{g_K}{h_K}$. Consider
two users indexed by $m$ and $n$ with $m<n$. Suppose at the optimal
solution,  $P^*_m<P^*_n$. Now, we show this assumption does not hold
by contradiction.

For convenience, we define $P_i^\prime\triangleq h_iP_i$. Then,
Problem \ref{Problem-Peakpeak2} can be rewritten as
\begin{pro}\label{Problem-PeakReformulation}
\begin{align}
\max_{\boldsymbol{P}}~~ &\sum_{i=1}^K \ln\left(1+\frac{
P_i^\prime}{\sigma^2+\sum_{j=1,j\neq i}^K P_j^\prime}\right),\label{Problem5-Obj}\\
\mbox{s.t.}~~ &P_i^\prime\ge 0, ~\forall i,\\
&P_i^\prime\le h_iP^{pk},~\forall i,\\
&\sum_{i=1}^K \frac{g_i}{h_i}P_i^\prime \le I.
\end{align}
\end{pro}

In Theorem \ref{Proposition-extreme}, we have proved that there is
at most one fractional user. Thus, the value of $P^*_m$ and $P^*_n$
has the following two cases.

Case 1: $0<P^*_m<P^{pk}$ and $P^*_{n}=P^{pk}$. It follows that
$P_m^\prime=h_mP^*_m$ and $P_{n}^\prime=h_{n}P^{pk}$. Then, based on
the relationship between $P_m^\prime$ and $P_n^\prime$, we have the
following two subcases:
\begin{itemize}
\item  {Subcase 1: $P_m^\prime<P_{n}^\prime$}. Now, we swap the power allocation of these two users, i.e.,
$\tilde{P}_m^\prime=h_{n}P^{pk}$ and
$\tilde{P}_{n}^\prime=h_mP^*_m$. Since $h_m>h_n$, it is clear that
$\tilde{P}_m^\prime=h_{n}P^{pk}<h_{m}P^{pk}$. Since
$P_m^\prime<P_{n}^\prime$, it is clear that
$\tilde{P}_{n}^\prime=h_mP^*_m<h_nP^{pk}$. At the same time, since
$\frac{g_m}{h_m}<\frac{g_n}{h_n}$, we have that
$\frac{g_m}{h_m}\tilde{P}_m^\prime+\frac{g_n}{h_n}\tilde{P}_n^\prime<\frac{g_m}{h_m}P_m^\prime+\frac{g_n}{h_n}P_n^\prime$.
Thus, the power allocation $(\tilde{P}_m^\prime,
\tilde{P}_n^\prime)$ is a feasible solution of Problem
\ref{Problem-PeakReformulation}.  Besides, it is observed that the
value of \eqref{Problem5-Obj} under $(\tilde{P}_m^\prime,
\tilde{P}_n^\prime)$ is the same as that under $(P_m^\prime,
P_n^\prime)$. Thus, $(\tilde{P}_m^\prime, \tilde{P}_n^\prime)$ is
also an optimal solution of Problem \ref{Problem-PeakReformulation}.
  \item {Subcase 2: $P_m^\prime>P_{n}^\prime$}. Now, we
consider the power allocation $\tilde{P}_m^\prime=P_m^\prime+\Delta$
and $\tilde{P}_{n}^\prime=P_{n}^\prime-\Delta$, where $\Delta$ is a
small constant such that $\tilde{P}_m^\prime\le h_mP^{pk}$ and
$\tilde{P}_{n}^\prime\ge 0$. Since
$\frac{g_m}{h_m}<\frac{g_n}{h_n}$, it is easy to verify that
$\frac{g_m}{h_m}\tilde{P}_m^\prime+\frac{g_n}{h_n}\tilde{P}_n^\prime<\frac{g_m}{h_m}P_m^\prime+\frac{g_n}{h_n}P_n^\prime$.
Thus, the power allocation $(\tilde{P}_m^\prime,
\tilde{P}_n^\prime)$ is a feasible solution of Problem
\ref{Problem-PeakReformulation}. Define \eqref{Problem5-Obj} as
$f\left(\boldsymbol{P}^\prime\right)$. It follows that
\begin{align}
f\left(\boldsymbol{P}^\prime\right)&=\ln\left(1+\frac{
P_m^\prime}{\sigma^2+\sum_{j\neq m, n}^K
P_j^\prime+P_n^\prime}\right)+\ln\left(1+\frac{
P_n^\prime}{\sigma^2+\sum_{j\neq m, n}^K
P_j^\prime+P_m^\prime}\right)\nonumber
\\&+\sum_{i=1}^K\ln\left(1+\frac{ P_i^\prime}{\sigma^2+\sum_{
j\neq i, m, n}^K P_j^\prime+P_m^\prime+P_n^\prime}\right).
\end{align}

\begin{align}
f\left(\tilde{\boldsymbol{P}}^\prime\right)&=\ln\left(1+\frac{
P_m^\prime+\Delta}{\sigma^2+\sum_{j\neq m, n}^K
P_j^\prime+P_n^\prime-\Delta}\right)+\ln\left(1+\frac{
P_n^\prime-\Delta}{\sigma^2+\sum_{j\neq m, n}^K
P_j^\prime+P_m^\prime+\Delta}\right)\nonumber
\\&+\sum_{i=1}^K\ln\left(1+\frac{ P_i^\prime}{\sigma^2+\sum_{
j\neq i, m, n}^K
P_j^\prime+P_m^\prime+\Delta+P_n^\prime-\Delta}\right).
\end{align}
For convenience, we define $Q^\prime=\sigma^2+\sum_{j\neq m, n}^K
P_j^\prime$.  Then, it follows that
\begin{align}
f\left({\boldsymbol{P}}^\prime\right)-f\left(\tilde{\boldsymbol{P}}^\prime\right)&=\ln\left(\frac{\left(Q^\prime+P_n^\prime-\Delta\right)\left(Q^\prime+P_m^\prime+\Delta\right)}{\left(Q^\prime+P_n^\prime\right)\left(Q^\prime+P_m^\prime\right)}\right)
\nonumber\\
&=\ln\left(\frac{\left(Q^\prime+P_n^\prime\right)\left(Q^\prime+P_m^\prime\right)-\left(P_m^\prime-P_n^\prime\right)\Delta-\Delta^2}{\left(Q^\prime+P_n^\prime\right)\left(Q^\prime+P_m^\prime\right)}\right)
\nonumber\\
&\stackrel{a}{\le} 0,
\end{align}
where ``a'' results from the fact that $\Delta >0$ and
$P_m^\prime>P_n^\prime$. This contradicts with our assumption.
\end{itemize}

Case 2: $P^*_m=0$ and $P^*_{n}>0$. It follows that $P_m^\prime=0$
and $P_{n}^\prime=h_{n}P^*_{n}$. We swap the power allocation of
these two users, i.e., $\tilde{P}_m^\prime=h_{n}P^*_{n}$ and
$\tilde{P}_{n}^\prime=0$. Since $h_m>h_n$, it is clear that
$\tilde{P}_m^\prime=h_{n}P^*_{n}<h_{m}P^{pk}$. At the same time,
since $\frac{g_m}{h_m}<\frac{g_n}{h_n}$, it can be verified that
$\frac{g_m}{h_m}\tilde{P}_m^\prime+\frac{g_n}{h_n}\tilde{P}_n^\prime<\frac{g_m}{h_m}P_m^\prime+\frac{g_n}{h_n}P_n^\prime$.
Thus, the power allocation $(\tilde{P}_m^\prime,
\tilde{P}_n^\prime)$ is a feasible solution of Problem
\ref{Problem-PeakReformulation}.  Besides, it is observed that the
value of \eqref{Problem5-Obj} under $(\tilde{P}_m^\prime,
\tilde{P}_n^\prime)$ is the same as that under $(P_m^\prime,
P_n^\prime)$.  Thus, $(\tilde{P}_m^\prime, \tilde{P}_n^\prime)$ is
also an optimal solution of Problem \ref{Problem-PeakReformulation}.

Thus, combining the results Case 1 and Case 2, it is clear that
there exists an optimal solution: for any two users indexed by $m$
and $n$, if $m<n$, their power allocation satisfies $P^*_m \ge
P^*_n$. Theorem \ref{Proposition-specialcase} is thus proved.
\end{proof}

Based on this theorem, we can develop the Algorithm
\ref{alg:specialcase} with linear complexity to solve Problem
\ref{Problem-Peakpeak2} when the users can be sorted in the order
stated in Theorem \ref{Proposition-specialcase}.

\begin{algorithm}[h!]
\caption{Optimal power allocation for Problem
\ref{Problem-Peakpeak2} with channel
ordering}\label{alg:specialcase}
\begin{algorithmic}[1]

\IF {$g_1P^{pk}>{Q}$}

\STATE $R^*=\log\left(1+\frac{h_1Q}{g_1\sigma^2}\right)$, $k^*=1$.

\ELSE

\STATE Initialize $k=1$.

\STATE Find the largest $k$ that satisfies $\sum_{i=1}^k
g_iP^{pk}\le I^{pk}$ and $k \le K$. Denote this $k$ as $k^L$.

\STATE Initialize
$R(1)=\log\left(1+\frac{h_1P^{pk}}{\sigma^2}\right)$, $k^*=1$,
$R^*=R(1)$.

\FOR {$k=2$ to $k^L$}

\STATE $R(k)=\sum_{i=1}^k \ln\left(1+\frac{
h_{i}P^{pk}}{\sigma^2+\sum_{j=1,j\neq i}^k h_{j}P^{pk}}\right)$.

\IF {$R(k)>R^*$}

\STATE $R^*=R(k)$, $k^*=k$.

\ENDIF

\ENDFOR

\STATE $R(k^L+1)=\sum_{i=1}^{k^L+1} \ln\left(1+\frac{
h_{i}P_{i}}{\sigma^2+\sum_{j=1,j\neq i}^{k^L+1} h_{j}P_{j}}\right)$,
where $P_{i}=P^{pk}, \forall i \le k^L$, and $P_{k^L+1}=
\frac{I^{pk}-\sum_{i=1}^{k^L}g_iP^{pk}}{g_{k^L+1}}$.

\IF {$R(k^L+1)>R^*$}

\STATE $R^*=R(k^L+1)$, $k^*=k^L+1$.

\ENDIF

\ENDIF


\end{algorithmic}
\end{algorithm}

%
%
%
%
%
%
%
%
%
%
%
%
%
%
%
%
%
%
%
%
%
%

\section{Numerical results}\label{Sec-simulation}
In this section, several numerical results are given to evaluate the
performances of the proposed studies. All the channels involved are
assumed to be Rayleigh fading, and thus the channel power gains are
exponentially distributed. Unless specifically stated, we assume the
mean of the channel power gains is one. The noise power $\sigma^2$
at SBS is also assumed to be 1. For convenience, the transmit power
constraint at each SU is assumed to be the same. The numerical
results presented here are obtained by taking average over $10000$
rounds simulations. In this section, we only provide the simulation
results for the ergodic sum-rate under peak TPC and peak IPC. No
simulation results for the ergodic sum-rate under average TPC and
average IPC are provided. This is due to the fact that we have shown
that the optimal power allocation for the ergodic sum-rate
with/without SIC under average TPC and average IPC is the same. As a
result, the simulation results for this case are exactly the same as
those shown in \cite{zhang2009ergodic}.


\subsection{Ergodic Sum-Rate with/without SIC}
First, we compare the ergodic sum-rate for the fading C-MAC
with/without SIC under different combinations of TPC and IPC. In
Fig. \ref{kx2SU}, Fig. \ref{kx5SU} and Fig. \ref{kx10SU}, we show
the results for the fading C-MAC with $K=2,5,10$, respectively. It
is observed from all the figures that the ergodic sum-rate with SIC
is always larger than that without SIC under the same TPC and IPC.
This verifies our result that the ergodic sum-rate with SIC is a
upper-bound of that without SIC. It is also observed that the gap
between the ergodic sum-rate with SIC and that without SIC in
general increases with the increasing of the number of SUs ($K$).
The engineering insight behind this is that when the number of SUs
is small in the C-MAC, it is not necessary to implement SIC at the
SBS due to the cost and complexity. While when the number of SUs is
large, it is worthwhile implementing SIC at the SBS to achieve a
larger sum-rate. It is observed from all the curves that when the
TPC of SU is large, the ergodic sum-rate gap with/without SIC is
negligible. This is due to the following fact. When TPC is very
large, TPC will not be the bottleneck, and the performance of the
C-MAC will only depend on the IPC. It is proved in \cite{XKang2013}
that the ergodic sum-rate with/without SIC under only the IPC is the
same. The optimal resource allocation for both cases are D-TDMA, and
let the SU with the best $h_i/g_i$ to transmit in each fading block.
Thus, from engineering design perspective, it is not necessary to
implement SIC at the SBS when the TPC is relatively large as
compared to the IPC.

\subsection{Optimality of the D-TDMA}

In Fig. \ref{kx1Sudomnumerical}, we numerically compute the
probability of D-TDMA being optimal for different number of SUs
based on the condition given in Lemma
\ref{Proposition-LargeNumebrSU}. First, it is observed that the
probability increases with the increasing of the number of SUs. It
is also observed that the probability increases with the increasing
of $P^{pk}$ for the same number of SUs. When $P^{pk}=5dB$ or $10dB$,
with only $10$ SUs, the probability of D-TDMA being optimal is close
to $1$. When $P^{pk}=0dB$, with $20$ SUs, the probability of D-TDMA
being optimal is more than $95\%$. These indicates that when the
number of SUs is sufficiently large, the D-TDMA is optimal with a
high probability. As pointed out previously, the condition given in
Lemma \ref{Proposition-LargeNumebrSU} is a sufficient condition. In
practice, the probability that D-TDMA is optimal is higher than
the probability shown in Fig. \ref{kx1Sudomnumerical}. 

In Fig. \ref{kx1Sudomsimu}, we compare the ergodic sum-rate under
the optimal power allocation and that under the D-TDMA when the
number of SUs is $5$. The results are obtained by averaging over
$10000$ rounds simulations. It is observed from the figure that when
the TPC is larger than $0dB$, D-TDMA can achieve the same ergodic
sum-rate as the optimal power allocation. Even when the TPC is less
than $0dB$, the gap between the ergodic sum-rate under the optimal
power allocation and that under the D-TDMA is not large. Thus, in
general, we can use the D-TDMA scheme as a good suboptimal scheme
when the number of SUs is larger than $5$.

\subsection{Ergodic Sum-Rate under the SIC-OP}
In this subsection, we compute the optimal power allocation for
C-MAC with SIC first, and then apply the obtained power allocation
to C-MAC without SIC (Problem \ref{Problem-Peakpeak}) as a
suboptimal power allocation. For convenience, we denote the optimal
power allocation for C-MAC with SIC as \emph{SIC-OP}. We then
compare the ergodic sum-rate (without SIC) under the SIC-OP with
that under the D-TDMA. The ergodic sum-rate (without SIC) under the
optimal power allocation are also included as a reference.

In Fig. \ref{kxSICOP1} and Fig. \ref{kxSICOP2}, we assume that there
are $5$ SUs in the network, and the IPC is assumed to be $0 dB$.  In
Fig.  \ref{kxSICOP1}, we assume that the mean of the channel power
gain is 1, i.e., $\mathbb{E}\{h_i\}=\mathbb{E}\{g_i\}=1, \forall i$.
It is observed from Fig. \ref{kxSICOP1} that there exists one
crossing-point, before which SIC-OP performs better than the D-TDMA.
Actually, SIC-OP can achieve the same performance as the optimal
power allocation when the TPC is sufficiently small. After the
crossing point, D-TDMA performs better than the SIC-OP. When the TPC
is sufficiently large, D-TDMA can achieve the same performance as
the optimal power allocation. Similar results can be observed in
Fig. \ref{kxSICOP2}, in which we assume that the mean of the channel
power gain is 0.1, i.e., $\mathbb{E}\{h_i\}=\mathbb{E}\{g_i\}=0.1,
\forall i$. The difference between Fig. \ref{kxSICOP2} and Fig.
\ref{kxSICOP1} is that the crossing point of Fig. \ref{kxSICOP2} has
a larger value of $P^{pk}$ as compared to the the crossing point of
Fig. \ref{kxSICOP1}. This can be explained as follows. According to
Lemma \ref{Proposition-LargeNumebrSU},  the condition for D-TDMA
being optimal is $\ln\left(1+h_{k}P_{k}^*/\sigma^2\right)\ge 1$.
Thus, when the mean of $h_{k}$ is small, a larger $P_{k}^*$ is
needed to make D-TDMA optimal.

In the following, we explain why SIC-OP can achieve the same
performance as the optimal power allocation when the TPC is small.
Now, we look at the sum-rate of MAC without SIC, which is
$\sum_{i=1}^K \ln\left(1+\frac{ h_{i}P_i}{\sigma^2+\sum_{j=1,j\neq
i}^K h_{j}P_j}\right)$. When $\frac{
h_{i}P_i}{\sigma^2+\sum_{j=1,j\neq i}^K h_{j}P_j}$ is small, it is
equivalent to $\sum_{i=1}^K \frac{
h_{i}P_i}{\sigma^2+\sum_{j=1,j\neq i}^K h_{j}P_j}$, since
$\ln(1+x)\thickapprox x$ when $x$ is small. Further, since the TPC
is small, $\sum_{j=1,j\neq i}^K h_{j}P_j\thickapprox\sum_{j=1}^K
h_{j}P_j$. Thus, $\sum_{i=1}^K \frac{
h_{i}P_i}{\sigma^2+\sum_{j=1,j\neq i}^K h_{j}P_j}\thickapprox
 \frac{ \sum_{i=1}^K h_{i}P_i}{\sigma^2+\sum_{j=1}^K
h_{j}P_j}=\frac{1}{1+\sigma^2/\sum_{i=1}^K h_{i}P_i}$. Thus,
maximizing $\sum_{i=1}^K \ln\left(1+\frac{
h_{i}P_i}{\sigma^2+\sum_{j=1,j\neq i}^K h_{j}P_j}\right)$ is
equivalent to maximizing $\sum_{i=1}^K h_{i}P_i$. The sum-rate of
MAC with SIC is obtained by maximizing $\ln\left(1+\sum_{i=1}^K
h_{i}P_i\right)$, which is also equivalent to  maximizing
$\sum_{i=1}^K h_{i}P_i$, since the log function is a monotonic
increasing function.

Now, we explain why this observation is important. With this
observation, we can solve Problem 3 by $\max\{\mbox{SIC-OP},
\mbox{D-TDMA}\}$, which achieve the same performance as the optimal
power allocation for most cases. Besides, the complexity is much
lower than searching the extreme points, especially when the number
of SUs is large.

\section{Conclusions}\label{conclusion}

In this paper, we studied the ergodic sum-rate of a spectrum-sharing
cognitive multiple access channel (C-MAC), where a secondary network
(SN) with multiple secondary users (SUs) shares the spectrum band
with a primary user (PU). We assumed an interference power
constraint at the PU, individual transmit power constraints at the
SUs, and to reduce decoding complexity, no successive interference
cancellation (SIC) at the C-MAC. We investigated the optimal power
allocation strategies for two types of power constraints: (1)
average TPC and average IPC, and (2) peak TPC and peak IPC. For the
average TPC and average IPC case, we proved that the optimal power
allocation is dynamic time-division multiple-access (D-TDMA). For
the peak TPC and peak IPC case, we proved that the optimal solution
must be at the extreme points of the feasible region. We showed that
D-TDMA is optimal with high probability when the number of SUs is
large. We also showed through simulations that the optimal power
allocation to maximize the ergodic sum-rate of the fading C-MAC with
SIC is optimal or near-optimal for our setting when D-TDMA is not
optimal. In addition, when some channel conditions are met, we gave
a linear time complexity algorithm for finding the optimal power
allocation.

\newpage
\begin{figure}[t]
        \centering
        \includegraphics*[width=12cm]{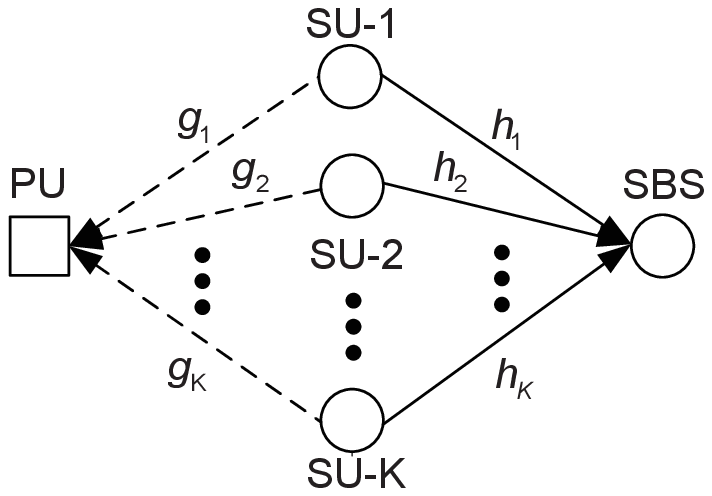}
        \caption{System model for a two-user fading C-MAC}
        \label{model}
\end{figure}


\begin{figure}[t]
        \centering
        \includegraphics*[width=12cm]{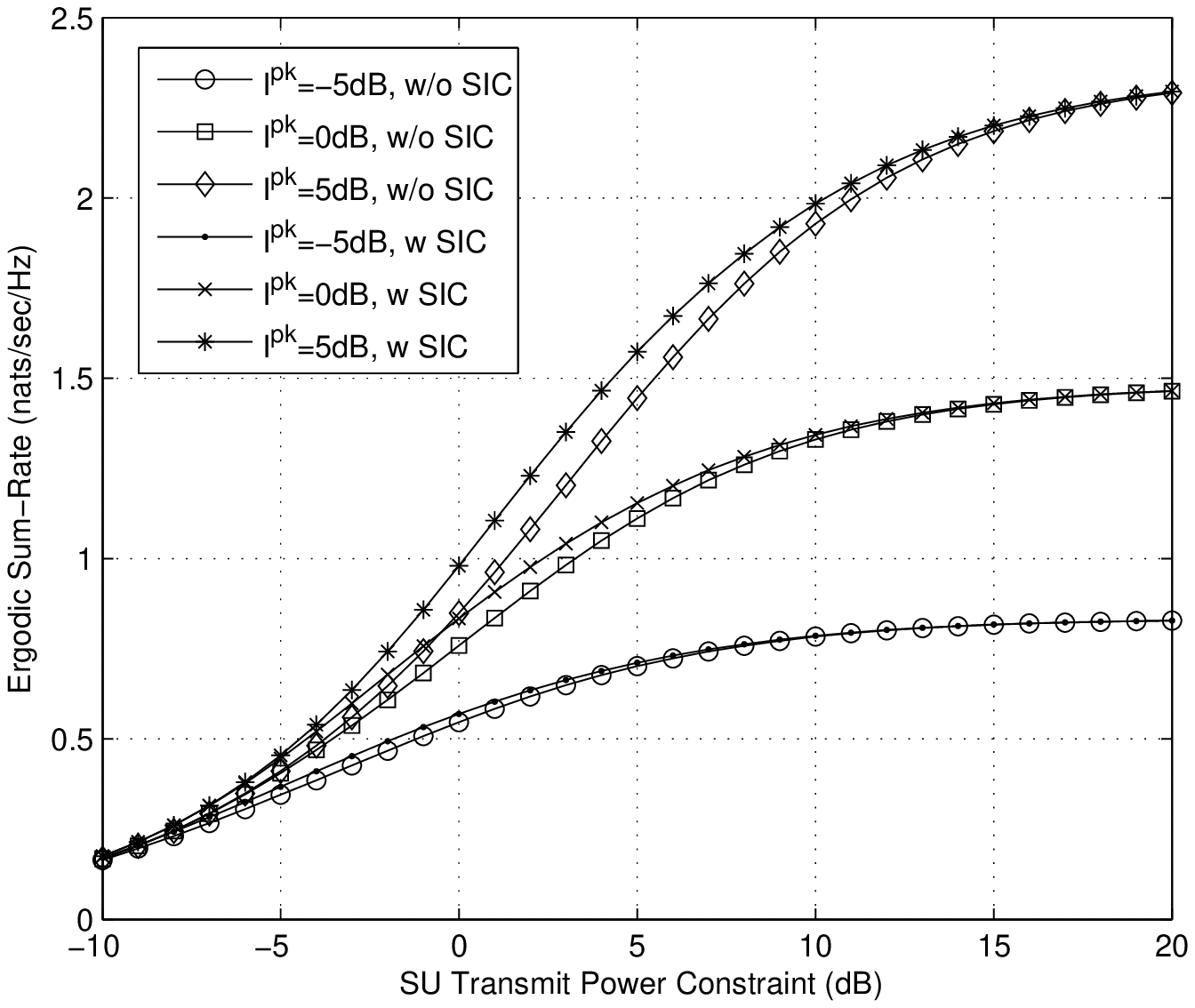}
        \caption{Ergodic Sum-Rate vs. the transmit power of SUs ($K=2$, $\sigma^2=1$, $\mathbb{E}\{h_i\}=\mathbb{E}\{g_i\}=1$)}
        \label{kx2SU}
\end{figure}

\begin{figure}[t]
        \centering
        \includegraphics*[width=12cm]{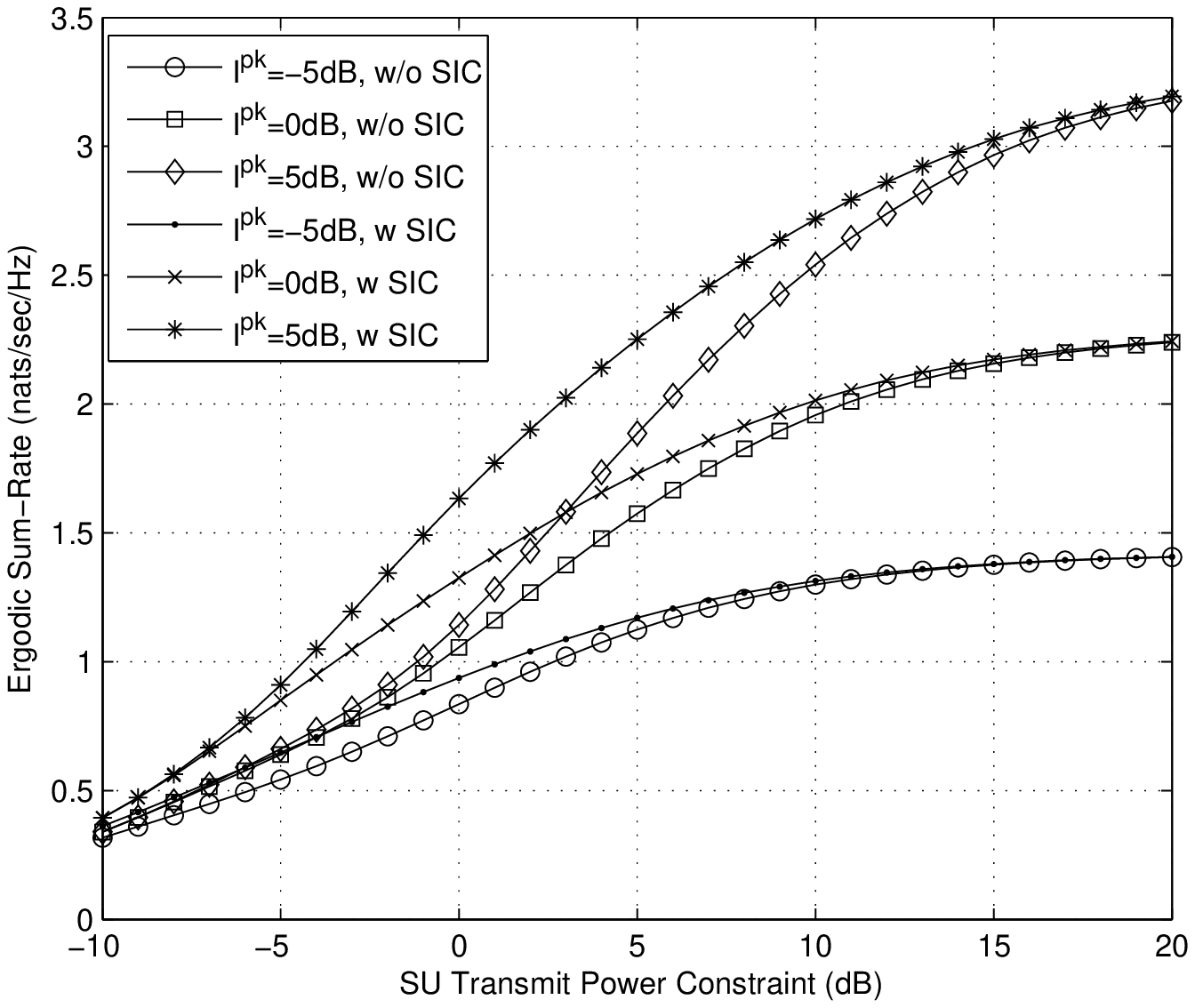}
        \caption{Ergodic Sum-Rate vs. the transmit power of SUs ($K=5$, $\sigma^2=1$, $\mathbb{E}\{h_i\}=\mathbb{E}\{g_i\}=1$)}
        \label{kx5SU}
\end{figure}

\begin{figure}[t]
        \centering
        \includegraphics*[width=12cm]{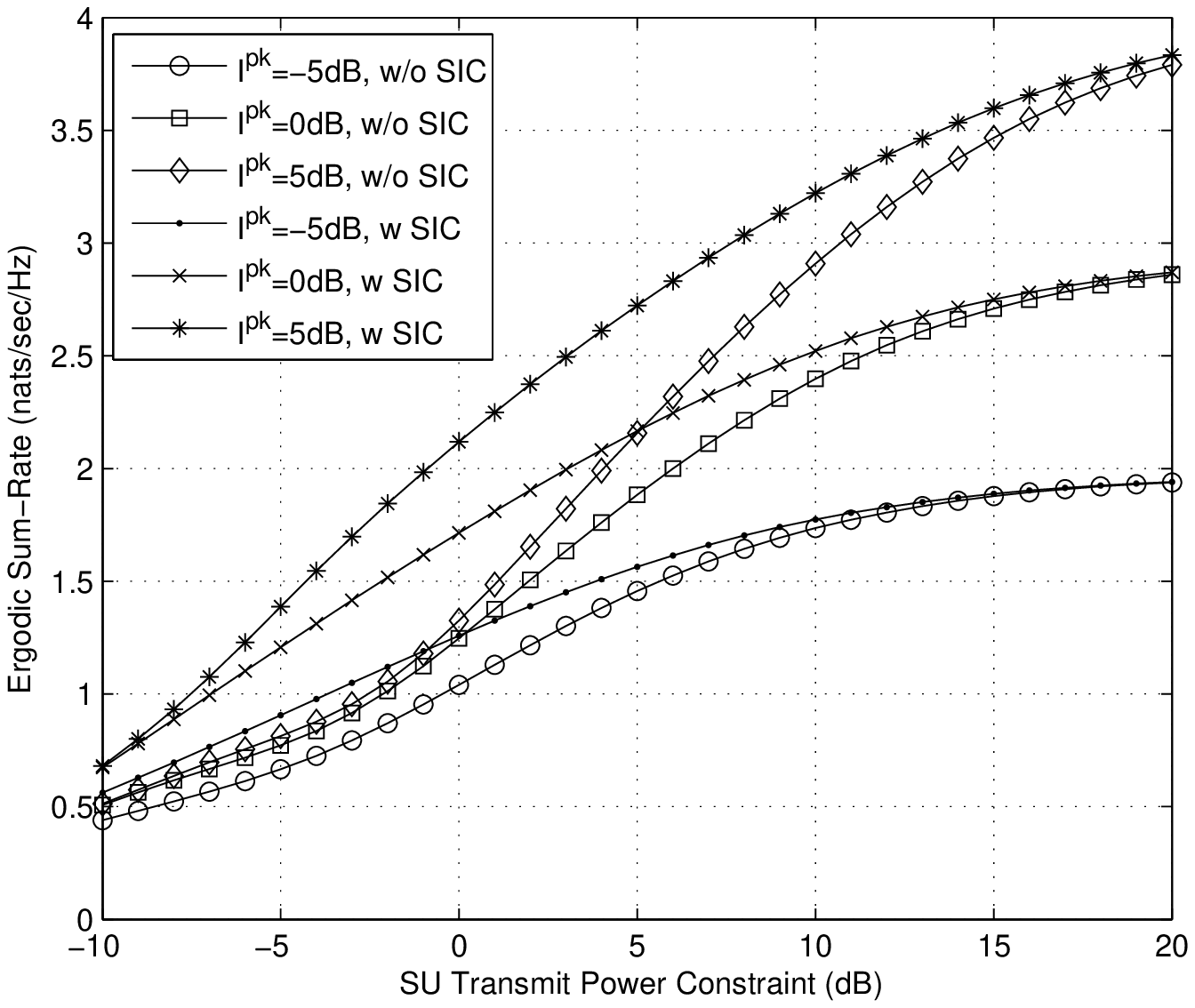}
        \caption{Ergodic Sum-Rate vs. the transmit power of SUs ($K=10$, $\sigma^2=1$, $\mathbb{E}\{h_i\}=\mathbb{E}\{g_i\}=1$)}
        \label{kx10SU}
\end{figure}

\begin{figure}[t]
        \centering
        \includegraphics*[width=12cm]{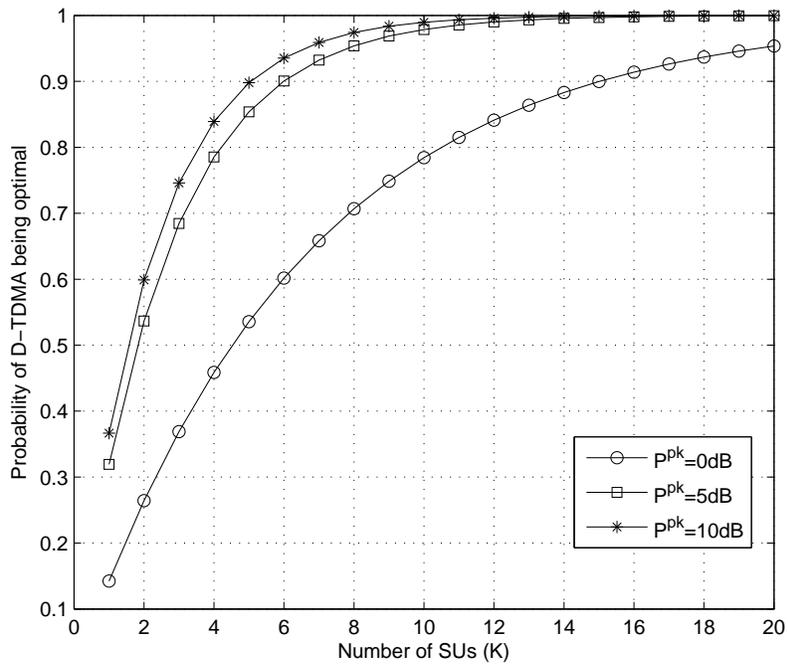}
        \caption{Probability of D-TDMA being optimal vs. the number of SUs ($I^{pk}=0 \mbox{dB}$, $\sigma^2=1$, $\mathbb{E}\{h_i\}=\mathbb{E}\{g_i\}=1$)}
        \label{kx1Sudomnumerical}
\end{figure}


\begin{figure}[t]
        \centering
        \includegraphics*[width=12cm]{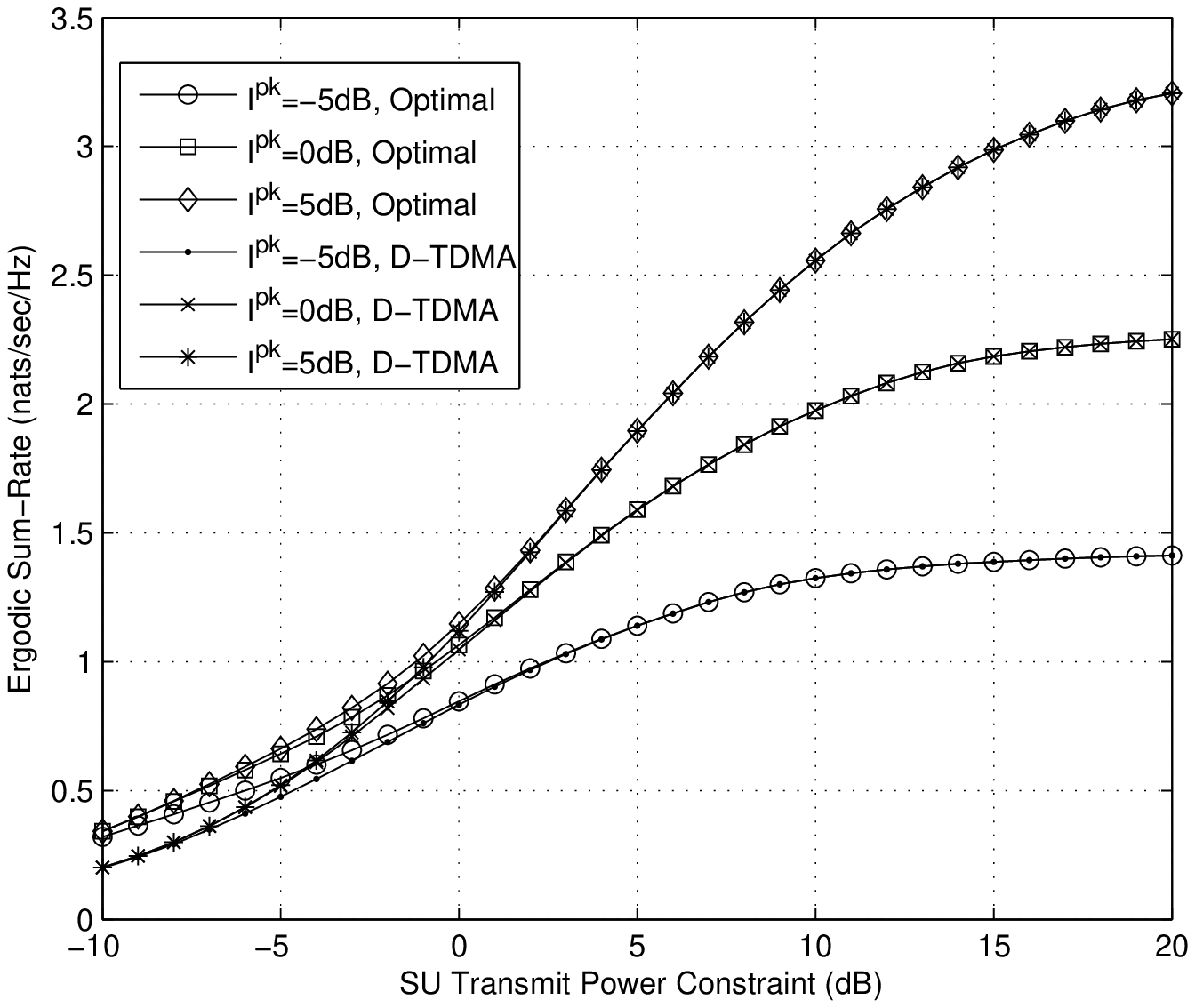}
        \caption{Comparison of the ergodic sum-rate: Optimal vs. D-TDMA  ($K=5$, $\sigma^2=1$, $I^{pk}=0 \mbox{dB}$, $\mathbb{E}\{h_i\}=\mathbb{E}\{g_i\}=1$)}
        \label{kx1Sudomsimu}
\end{figure}

\begin{figure}[t]
        \centering
        \includegraphics*[width=11.5cm]{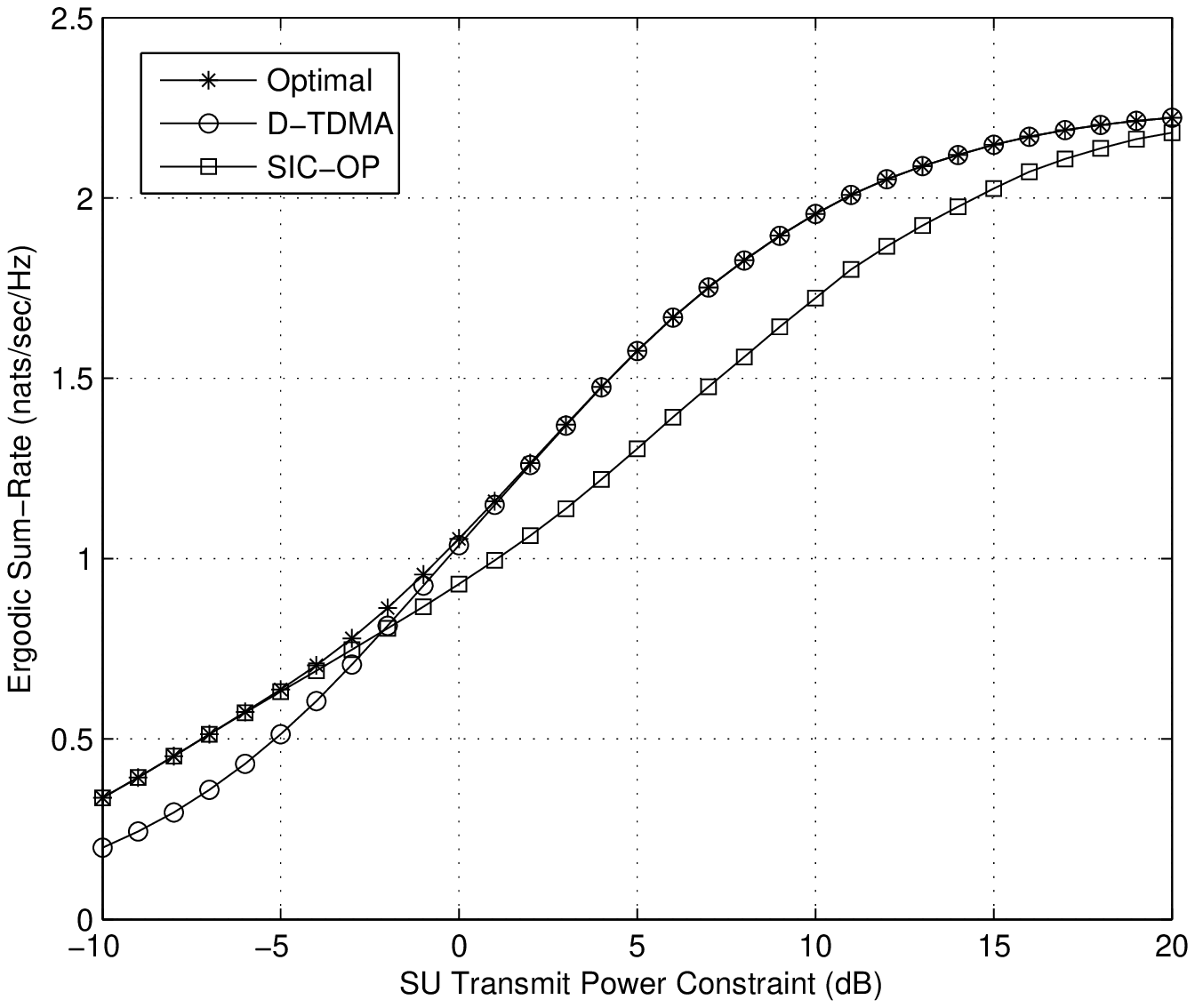}
        \caption{Comparison of the ergodic sum-rate: Optimal vs. D-TDMA vs. SIC-OP  ($K=5$, $\sigma^2=1$, $I^{pk}=0 \mbox{dB}$,  $\mathbb{E}\{h_i\}=\mathbb{E}\{g_i\}=1$)}
        \label{kxSICOP1}
\end{figure}

\begin{figure}[t]
        \centering
        \includegraphics*[width=11.5cm]{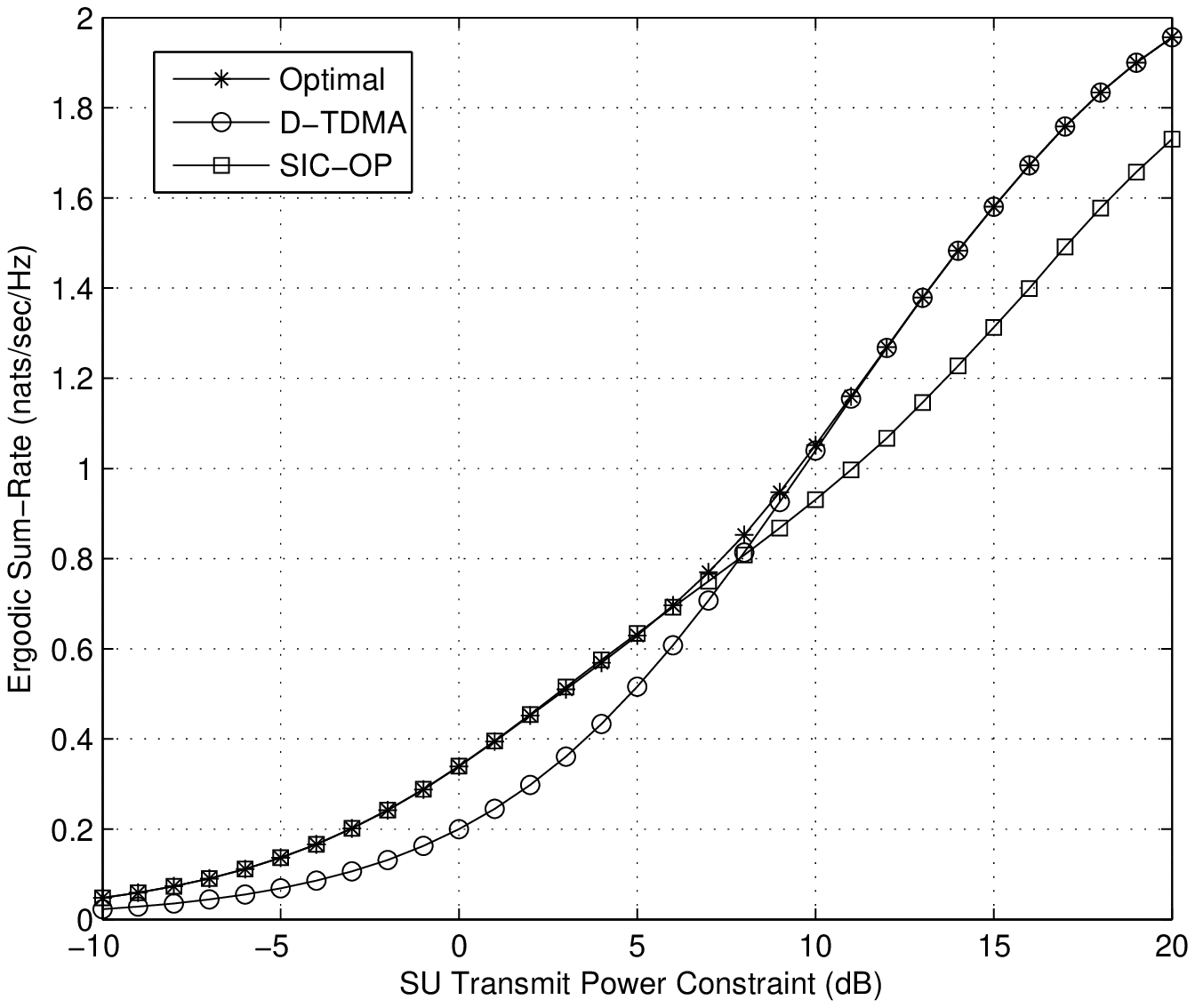}
        \caption{Comparison of the ergodic sum-rate: Optimal vs. D-TDMA vs. SIC-OP  ($K=5$, $\sigma^2=1$, $I^{pk}=0 \mbox{dB}$, $\mathbb{E}\{h_i\}=\mathbb{E}\{g_i\}=0.1$)}
        \label{kxSICOP2}
\end{figure}

\end{document}